\documentclass{comjnl}

\usepackage{amsmath,amssymb}
\usepackage{marvosym}

\usepackage{threeparttable}

\usepackage{booktabs}
\usepackage{array}
\makeatletter
\newcommand{\thickhline}{%
    \noalign {\ifnum 0=`}\fi \hrule height 0.8pt
    \futurelet \reserved@a \@xhline
}
\newcolumntype{"}{@{\hskip\tabcolsep\vrule width 1pt\hskip\tabcolsep}}
\makeatother

\begin{document}

\title[An Improved Lattice-Based Ring Signature with Unclaimable Anonymity in the Standard Model]{An Improved Lattice-Based Ring Signature with Unclaimable Anonymity in the Standard Model}

\author{Mingxing Hu}
\affiliation{Department of Computer Science and Engineering, Shanghai Jiao Tong University, Shanghai 200240, China} 

\author{Weijiong Zhang}
\affiliation{Shanghai Electro-Mechanical Engineering Institute, Shanghai 201109, China}

\author{Zhen Liu}
\affiliation{Department of Computer Science and Engineering, Shanghai Jiao Tong University, Shanghai 200240, China} 
\email{mxhu2018@sjtu.edu.cn, billcat2021@sina.com, liuzhen@sjtu.edu.cn}

\shortauthors{M. Hu \emph{et al.}}

%\received{00 January 2009}
%\revised{00 Month 2009}

%\category{C.2}{Computer Communication Networks}{Computer Networks}
%\category{C.4}{Performance of Systems}{Analytical Models}
%\category{G.3}{Stochastic Processes}{Queueing Systems}
%\terms{Internet Technologies, E-Commerce}
\keywords{Lattice-Based; Ring Signature; Standard Model; Unclaimable Anonymity}

\begin{abstract}
%Ring signatures enable a user to sign messages on behalf of an arbitrary set of users, called the ring, without revealing exactly which member of that ring actually generated the signature. The signer-anonymity property makes ring signatures have been an active research topic. Anonymity w.r.t. full key exposure is the strongest privacy notion in the hierarchy of privacy notions presented by Bender et al. (TCC’06) and is widely used in many researches. This privacy notion allows the adversary to compromise all the randomness that is used to produce the signing keys of all ring members. Recently, Park and Sealfon (Crypto’19) presented a stronger notion named \emph{signer-unclaimability} which not only allows the adversary to compromise all the randomness that is used to produce the signing keys but also the signing randomness of all the ring members. This work also proposed a ring signature scheme with unclaimable anonymity in the standard model, however, it did not consider the unforgeable w.r.t. adversarially-chosen-key attack (the public key ring of a signature may contain keys created by an adversary) and the signature size grows quadratically in the size of ring and message.

%In this work, we propose a new lattice-based ring signature scheme with unclaimable anonymity in the standard model. In particular, our work improves the security and efficiency of Park and Sealfon's work, which is unforgeable w.r.t. adversarially-chosen-key attack, and the ring signature size grows linearly in the ring size. 
Ring signatures enable a user to sign messages on behalf of an arbitrary set of users, called the ring, without revealing exactly which member of that ring actually generated the signature. The signer-anonymity property makes ring signatures have been an active research topic. Recently, Park and Sealfon (CRYPTO’19) presented an important anonymity notion named \emph{signer-unclaimability} and constructed a lattice-based ring signature scheme with unclaimable anonymity in the standard model, however, it did not consider the unforgeable w.r.t. adversarially-chosen-key attack (the public key ring of a signature may contain keys created by an adversary) and the signature size grows quadratically in the size of ring and message.

In this work, we propose a new lattice-based ring signature scheme with unclaimable anonymity in the standard model. In particular, our work improves the security and efficiency of Park and Sealfon's work, which is unforgeable w.r.t. adversarially-chosen-key attack, and the ring signature size grows linearly in the ring size. 

\end{abstract}

\maketitle

\section{Introduction}
Ring signatures, introduced by Rivest et al. \cite{RST01}, allow a signer to hide in a \textit{ring} of potential signers of which the signer is a member, without revealing which member actually produced the signature. Thereafter, ring signatures have been researched extensively \cite{AOS02,BSS02,Nao02,ZK02,BGL+03,HS03,DKN+04,LWW04,CWL+06,BKM06}. Among these works, it is worth mentioning that Bender et al. \cite{BKM06} presented a hierarchy of security and privacy models which were widely used in many works \cite{BK10,LLN+16,MS17,TSS+18,BHK+18,TKS+19,PS19,BDH+19,CGH+21}. Among these privacy models, the strongest one is anonymity w.r.t. full key exposure which allows that even if an adversary compromises the randomness used to produce these signing keys of all the ring members in a ring, the adversary cannot identify the real signers of past signatures. Recently, Park and Sealfon \cite{PS19} presented a stronger privacy notion named \textit{signer-unclaimability} and showed that signer-unclaimability implies the signer-anonymity w.r.t. full key exposure. The signer-unclaimability notion not only allows the adversary to compromise all the randomness used to produce the signing keys but also all the randomness used to produce the signatures. 

%In other words, ring signatures with unclaimable anonymity captures the case that existing a ring member who attempts to collaborate with attacker, even the ring member reveals his randomness used to produce the signing key and the randomness used to produce the signature()

Another important line of research is the ring signature constructions from lattices \cite{BK10,PS19,CGH+21,WS11,MBB+13,NCJ14,ZHX+16,WZZ18,ESS+19,ESL+19,EZS+19,LAZ19,LYJ+19,MYV+20}, since lattice-based cryptography has attracted more attention due to its distinctive features especially quantum resistant. However, the security of the majority of these works relies on random oracle (ROM) heuristic. As shown by Katz \cite{Kat10} (Sect. 6.2.1), it will arouse concerns on the basic security of cryptosystems that rely on ROM. For instance, Leurent and Nguyen \cite{LN09} presented the attacks extracting the secret keys on several hash-then-sign type signature schemes (including the lattice-based signature \cite{GPV08}) when the underlying hash functions are modeled as random oracle. The first lattice-based ring signature in the standard model proposed by Brakerski and Tauman-Kalai \cite{BK10}, however, it did not consider the security notion i.e., unforgeable w.r.t. \emph{adversarially-chosen-key attack}, and the signature size grows quadratically in the ring and message size. Recently, Park and Sealfon \cite{PS19} adapted the work \cite{BK10} to a ring signature with unclaimable anonymity, which retained the merits of the standard model and standard lattice assumption, but was still not unforgeable w.r.t. adversarially-chosen keys attack and the signature size is still large. 

%and with the rapid development of cryptocurrencies in recent years, ring signature has been an active research topic and many works has been proposed \cite{LLN+16,MS17,TSS+18,BHK+18,TKS+19,ESS+19,ESL+19,PS19,EZS+19,BDH+19,LNP+20,CGH+21}.
%

Overcoming these two weaknesses is very necessary for practical scenarios. Particularly, in cryptocurrencies, an attacker can create some public keys maliciously and put them into the blockchain as the normal ones, in this case, honest users may include these maliciously created keys in their rings to sign their transactions. Therefore, the security model must consider the attack in such a scenario, which is referred to as \emph{adversarially-chosen-key attack}. Moreover, the signature size can not grow too fast with the ring size since the number of members in cryptocurrencies is usually huge.

\subsection{Our Results}\label{OurResults}
To address the above concerns, we propose a new lattice-based ring signature scheme with unclaimable anonymity in the standard model based on standard lattice assumptions (SIS and LWE). In particular, our work \textit{simultaneously improves the security and efficiency} of Park and Sealfon's work \cite{PS19}. 

On the security, our ring signature is unforgeable w.r.t. adversarially-chosen-key attack, which is a stronger security notion than the one used in Park and Sealfon’s work \cite{PS19}; On the efficiency, we eliminate the dependency between the message length and ring signature size, i.e., make the ring signature size grows linearly in the ring size.

Table \ref{Table1} shows a comparison between our results in this work and the existing works on lattice-based ring signatures. We note that although the works \cite{ESS+19,ESL+19,EZS+19} achieved sublinear-size ring signatures, they rely on the random oracle heuristic and the anonymity is not unclaimable. And even though the work \cite{CGH+21} achieved logarithmic-size, their anonymity is not unclaimable and the construction employs many cost cryptographic building blocks and proof system which would likely render concrete instantiations inefficient for reasonable parameters.

\begin{table*}
\caption{Comparison with existing lattice-based ring signature schemes} 
\label{Table1}
\begin{threeparttable}
\vspace{0cm}
\begin{tabular}{l|l|l|l|l}
\thickhline
\begin{tabular}[c]{@{}l@{}}Lattice-based\\ring signature works \end{tabular}&\begin{tabular}[c]{@{}l@{}}Standard model~\end{tabular} &\begin{tabular}[c]{@{}l@{}}Unclaimable anonymity~\end{tabular} & \begin{tabular}[c]{@{}l@{}}Unforgeability w.r.t.\\adversarially-chosen-key attack\end{tabular}  & \begin{tabular}[c]{@{}l@{}}Signature size~~~~\end{tabular}   \\ \thickhline

\begin{tabular}[c]{@{}l@{}}\cite{WS11,NCJ14}\end{tabular}   
&$\surd$         &$\times$    	     &$\surd$, flaws\tnote{a}                  &Linear \\ \hline
                      
\begin{tabular}[c]{@{}l@{}}\cite{MBB+13,ZHX+16,WZZ18,LAZ19,LYJ+19,MYV+20}~\\\end{tabular} 
&$\times$        &$\times$           &$\surd$  			 			&Linear \\ \hline
                                   
\begin{tabular}[c]{@{}l@{}}\cite{ESS+19,ESL+19,EZS+19}\end{tabular}     
 &$\times$       &$\times$          &$\surd$                   			&Sub-linear  \\ \hline
                        
\cite{CGH+21}   
&$\surd$         &$\times$           &$\surd$                    			&Logrithmic   \\ \hline
 
\begin{tabular}[c]{@{}l@{}}\cite{BK10,PS19}\end{tabular}
 &$\surd$        &$\surd$            &$\times$                    			&Quadratic   \\ \hline

Our                  
&$\surd$         &$\surd$  	    &$\surd$                      			&Linear        \\ \thickhline
\end{tabular}

\noindent \begin{tablenotes}
        \footnotesize
        \item[a] The works \cite{WS11,NCJ14} are not unforgeable w.r.t. adversarially-chosen-key attack since the forger can trivially forge a signature by querying the signing oracle with his adversarially-chosen public keys ability (Refer to Sect. \ref{Approach} for the details of the attack). 
      \end{tablenotes}
    \end{threeparttable}
\vspace{-0.7cm}
\end{table*}
%Moreover, our methods only employ one specific cryptographic building block i.e., a pseudorandom function (PRF) proposed by Lai et al. \cite{LLW20}, rather than employ many general cryptographic building blocks and NIWI proof system as \cite{BDH+19} which would likely render concrete instantiations inefficient for reasonable parameters.

\subsection{Overview of Our Approach} \label{Approach}

To describe our approach, it is instructive to recall Park and Sealfon's work \cite{PS19}. In the PS scheme, there are $N$ users in a ring $\textsf{R}$, each user generates $\bold A^{(i)}$ and trapdoor $\bold S^{(i)}$ by trapdoor generation algorithm, and samples $2t$ ``message matching'' matrices $\{\bold A_{j,b}^{(i)}\}_{(j,b)\in[t]\times \{0,1\}}\xleftarrow{\$} \mathbb{Z}_q^{n\times m}$ two of them corresponding to each bit of the message. It additionally sample a vector $\bold y^{(i)}\in \mathbb{Z}_q^{n}$. Each user's verification key is $\textsf{vk}^{(i)}=(\bold A^{(i)},\{\bold A_{j,b}^{(i)}\}_{(j,b)\in[t]\times \{0,1\}},\bold y^{(i)})$, signing key is $\textsf{sk}^{(i)}=\bold S^{(i)}$. The ring $\textsf{R}=\{\textsf{vk}^{(1)},\dots,\textsf{vk}^{(N)}\}$. Let $\boldsymbol{\mu}=(\mu_1,\dots,\mu_t)\in \{0,1\}^t$ be the message. The signing procedure is one of the $N$ users who samples $\bold x^{(1)},\dots,\bold x^{(N)}\in \mathbb{Z}^{(t+1)\times m}$ such that the ring equation $\bar{\bold A}^{(1)}\bold x^{(1)} +\dots+ \bar{\bold A}^{(N)}\bold x^{(N)} = \bold y \pmod q$ holds, where $\bar{\bold A}^{(i)}=[\bold A^{(i)}| \bold A_{\boldsymbol{\mu}}^{(i)}]$, $\bold A_{\boldsymbol{\mu}}^{(i)}=[\bold A_{1,\mu_1}^{(i)}|\dots|\bold A_{t, \mu_t}^{(i)}]$, and $\bold y\leftarrow \{\bold y^{(i)}\}_{i\in[N]}$ selected in lexicographically first way. Finally, output $(\bold x^{(1)},\dots,\bold x^{(N)})$ as the signature. The verification procedure check if the signature is well-formed and if the ring equation holds, accept, otherwise reject.

In this setting, an adversary can easily forge a signature by querying the signing oracle with an adversarially-formed ring such as $\textsf{R}^*=\big\{\textsf{vk}^{(1)},\textsf{vk}^{(2)},c_1\textsf{vk}^{(1)},c_2\textsf{vk}^{(2)}\big\}$ where $c_1, c_2$ are constants and 
\begin{gather*}
c_1\textsf{vk}^{(1)}=\Big(c_1\bold A^{(1)},\big\{c_1\bold A_{j,b}^{(1)}\big\}_{(j,b)\in[t]\times \{0,1\}},\bold y^{(1)}\Big)\\
c_2\textsf{vk}^{(2)}=\Big(c_2\bold A^{(2)},\big\{c_2\bold A_{j,b}^{(2)}\big\}_{(j,b)\in[t]\times \{0,1\}},\bold y^{(2)}\Big) 
\end{gather*}
Assume $(\bold x^{(1)},\bold x^{(2)},\bold x^{(3)}, \bold x^{(4)})$ is the replied signature for that query, $\boldsymbol{\mu}$ is the queried message, then the adversary can immediately obtain a forgery signature $(\bold x^{(1)}+c_1\bold x^{(3)},\bold x^{(2)}+c_2\bold x^{(4)})$ for the ring $(\textsf{vk}^{(1)},\textsf{vk}^{(2)})$ since the ring equation $\bar{\bold A}^{(1)}(\bold x^{(1)}+c_1\bold x^{(3)}) + \bar{\bold A}^{(2)}(\bold x^{(2)}+c_2\bold x^{(4)}) = \bold y \pmod q$ holds for $\bold y\leftarrow \{\bold y^{(1)},\bold y^{(2)}\}$.

To resolve the above problem, we use the key-homomorphic evaluation algorithm that developed from \cite{GSW13,BV14,BGG+14} to evaluate circuits of a $\textsf{PRF}$. Even though the method is inspired by the standard signature work \cite{BL16}, it is essentially different in ring signature setting, since its privacy and security requirements is more complex than standard signature. For our construction, we borrow the idea from the unforgeability simulation of \cite{BL16}. In our setting, the verification key of each user is $\textsf{vk}^{(i)}=\big(\bold A^{(i)},(\bold A_0^{(i)},\bold A_1^{(i)}),\{\bold B_{j}^{(i)}\}_{j\in[k]},(\bold C_0^{(i)},\bold C_1^{(i)})\big)$ and the signing key is $\textsf{sk}^{(i)}=(\bold S^{(i)}, \bold k^{(i)})$, where $(\bold A^{(i)},\bold S^{(i)})$ are generated by trapdoor generation algorithm, $\bold k^{(i)}\in\{0,1\}^k$ is a \textsf{PRF} key, the remained matrices are used to construct the homomorphic evaluated matrix $\bold A_{C_\textsf{PRF},\boldsymbol{\mu}}^{(i)}$ in signing phase.

In the signing phase, the signer first construct a homomorphic evaluated matrix $\bold A_{C_\textsf{PRF},\boldsymbol{\mu}}^{(i)}$ which determined by the signer's $\textsf{PRF}$ key, where ${C_\textsf{PRF}}$ denote \textsf{NAND} Boolean circuit expression of the PRF function. Then sample a $2Nm$-dimensional vector $\bold x'=\big((\bold x^{(1)})^\top|\dots|(\bold x^{N})^\top\big)$ such that the following ring equation holds
\begin{equation}
\sum_{i\in[N]}\big[\bold A^{(i)}|\bold A_{C_\textsf{PRF},\boldsymbol{\mu}}^{(i)}\big]\bold x^{(i)}=\bold 0\pmod q
\label{Eq1}
\end{equation}
Finally, output $\bold x'$ as the signature. In the verification phase, check if the input signature is well-formed and the ring equation (\ref{Eq1}) holds.

 In this setting, it is effective against the adversary to forge signatures by adversarially forming a ring. Because $\bold A_{C_\textsf{PRF},\boldsymbol{\mu}}^{(i)}$ can not be predetermined i.e., is unpredictable for the adversary since $\bold A_{C_\textsf{PRF},\boldsymbol{\mu}}^{(i)}$ generated based on the $\textsf{PRF}$ key that selected during signing phase. For the ring signature size, we process the message by puncturing it into specific matrices and homomorphically evaluate them using $C_\textsf{PRF}$ which finally outputs only one matrix as result, rather than directly assigning a matrix to each message bit and finally outputs a concatenation of these matrices as \cite{PS19}, which eliminate the dependency between the message length and ring signature size and therefore our ring signature size is only linear with ring size.

\section{Definitions}\label{OurDefinitions}
In this section, we define the ring signature system and then formalize the security and privacy models.

\subsection{Algorithm Definition}\label{OurDefinitions}

\begin{definition}[Ring Signature]
A ring signature {\rm \textsf{RS}} scheme consists of the following algorithms: \label{RS}
\begin{itemize}
    \item {\rm $\textsf{Setup}(1^n)\rightarrow \textsf{PP}$}. This is a probabilistic algorithm. On input a security parameter $n$, the algorithm outputs the public parameters {\rm $\textsf{PP}$}.
    
    {\rm \textsf{The public parameters} $\textsf{PP}$ \textsf{are common parameters used by all participants in the system, for example, the message space $\mathcal{M}$, the signature space, etc.} $\textsf{In the following,}$ $\textsf{PP}$ $\textsf{are implicit input parameters to every algorithm.}$}
    
    \item {\rm $\textsf{KeyGen}()\rightarrow (\textsf{vk,sk})$}. This is a probabilistic algorithm. The algorithm outputs a verification key {\rm $\textsf{vk}$} and a signing key {\rm $\textsf{sk}$}.
    
    {\rm \textsf{Any ring member can run this algorithm to generate a pair of verification keys and signing keys.}}
    
    \item {\rm $\textsf{Sign}(\mu,\textsf{R},\textsf{sk})\rightarrow \Sigma$}. This is a probabilistic algorithm. On input a message $\mu\in \mathcal{M}$, a ring of verification keys {\rm $\textsf{R}=(\textsf{vk}^{(1)},\dots,\textsf{vk}^{(N)})$}\footnote{Below we regard the verification key ring as an ordered set, namely, it consists of a set of verification keys, and when it is used in \textsf{Sign} and \textsf{Ver} algorithms, the verification keys are ordered and each one has an index.} and a signing key {\rm \textsf{sk}}. Assume that (1) the input signing key {\rm \textsf{sk}} and the corresponding verification key {\rm \textsf{vk}} is a valid key pair output by {\rm\textsf{KeyGen}} and {\rm \textsf{vk}} $ \in $ {\rm\textsf{R}}, (2) the ring size $|${\rm\textsf{R}}$|\geq 2$, (3) each verification key in ring {\rm$\textsf{R}$} is distinct. The algorithm outputs a signature $\Sigma$.
    
    \item {\rm $\textsf{Ver}(\mu,\textsf{R},\Sigma)\rightarrow 1/0$}. This is a deterministic algorithm. On input a message $\mu$, a ring of verification keys {\rm $\textsf{R}=(\textsf{vk}^{(1)},\dots,\textsf{vk}^{(N)})$} and a signature $\Sigma$, the algorithm outputs 1 if the signature is valid, or 0 if the signature is invalid.
\end{itemize}
\end{definition}

\noindent\emph{Remark:} Note that it is open on whether the \textsf{Sign} algorithm is probabilistic or deterministic, which may depend on the concrete construction.

\noindent$\textbf{Correctness}$. \textit{A {\rm $\textsf{RS}$} scheme is correct, if for all $n\in \mathbb{N}$, all $N=\mathrm{poly}(n)$, all $i\in[N]$, all messages $\mu\in \mathcal{M}$, any {\rm $\textsf{PP}\leftarrow \textsf{Setup}(1^n)$} as implicit input parameter to every algorithm, any $N$ pairs {\rm $\{\textsf{vk}^{(i)},\textsf{sk}^{(i)}\}_{i\in[N]}\leftarrow \textsf{KeyGen}()$} and any {\rm $\Sigma\leftarrow\textsf{Sign}(\mu,\textsf{R},\textsf{sk}^{(i)})$} where {\rm $\textsf{R}=\{\textsf{vk}^{(1)},\dots,\textsf{vk}^{(N)}\}$}, it holds that}
\begin{center}
    {\rm $\mathrm{Pr}\big[\textsf{Ver}(\mu,\textsf{R},\Sigma)=1\big]=1-\mathrm{negl}(n)$}
\end{center}
\textit{where the probability is taken over the random coins used by {\rm \textsf{Setup}}, {\rm \textsf{KeyGen}}, and {\rm \textsf{Sign}}.}

\subsection{Security and Privacy Models}\label{OurDefinitions}

Below we define the security and privacy models for RS. In both models, we give the randomness used in {\rm $\textsf{Setup}$} to the adversary, which implies the {\rm $\textsf{Setup}$} algorithm is public, does not rely on a trusted setup that may incur concerns on the existence of trapdoors hidden in the output parameters. 

The security model i.e., unforgeability w.r.t. adversarially-chosen-key attack captures that only the ring member knowing the secret key for some verification key in a ring can generate a valid signature with respect to the ring, even though existing adversary is allowed to arbitrarily add some verification keys in the ring when querying the signing oracle. In other words, assuming there is a ring signature system that satisfies unforgeability w.r.t. adversarially-chosen-key attack, even some verification keys in the system were maliciously generated by an adversary, and these keys were used to issue signatures by some honest ring members, the unforgeability still holds.

The privacy model i.e., signer-unclaimability captures that given a valid signature with respect to a ring of verification keys and the randomness used to produce the signature, even though existing adversary obtained all the randomness that used to produce all the signing keys in the system, the adversary still can not identify the signer’s verification key out of the ring. In other words, assuming an RS system with unclaimable anonymity revealed all ring members' signing keys, and allowed the adversary to obtain the singing randomness of any valid signature in the system, the signer-anonymity still holds.

\noindent$\textbf{Unforgeability}$. \textit{A {\rm $\textsf{RS}$} scheme is unforgeability against adversarially-chosen-key attack {\rm (\textsf{UnfAdvKey})}, if for any PPT adversary $\mathcal{A}$, it holds that $\mathcal{A}$ has at most negligible advantage in the following experiment with a challenger $\mathcal{C}$.}
\begin{itemize}
    \item \textit{{\rm $\textbf{Setup.}$} $\mathcal{C}$ generates {\rm $\textsf{PP}\leftarrow \textsf{Setup}(1^n;\gamma_{\textsf{st}})$} and {\rm $(\textsf{vk}^{(i)},\textsf{sk}^{(i)})\leftarrow \textsf{KeyGen}()$} for all $i\in[N]$, where $N=\mathrm{poly}(n)$ and $\gamma_{\textsf{st}}$ is the randomness used in {\rm $\textsf{Setup}$}. $\mathcal{C}$ sets {\rm $\textsf{S}=\{\textsf{vk}^{(i)}\}_{i\in[N]}$} and initializes an empty set {\rm\textsf{L}}. Finally, $\mathcal{C}$ sends {\rm ($\textsf{PP}$, $\textsf{S},\gamma_{\textsf{st}}$)} to $\mathcal{A}$.}
    
%    \textsf{Note that we give $\mathcal{A}$ the randomness $\gamma_{\textsf{st}}$, which implies the Setup algorithm is public, does not rely on a trusted setup that may incur concerns on the existence of trapdoors hidden in the output parameters.}
    
    \item \textit{{\rm $\textbf{Probing Phase.}$} $\mathcal{A}$ is given access to a }\textit{signing oracle {\rm $\textsf{OSign}(\cdot,\cdot,\cdot)$}:}\textit{On input a message $\mu\in\mathcal{M}$, a ring of verification keys {\rm \textsf{R}} and an index $s\in[N]$ such that {\rm $\textsf{vk}^{(s)}\in \textsf{R}\bigcap \textsf{S}$}, this oracle returns {\rm $\Sigma\leftarrow\textsf{Sign}(\mu,\textsf{R},\textsf{sk}^{(s)})$} and adds the tuple {\rm$(\mu,\textsf{R},\Sigma)$} to {\rm\textsf{L}}.}
    
     \textsf{Note that it only requires that the $\textsf{vk}^{(s)}$ is in \textsf{S} without requiring that $\textsf{R}\subseteq \textsf{S}$. This captures that $\mathcal{A}$ can obtain the signing oracle of its choice, in which the queried ring \textsf{R} may contain verification keys that are created by $\mathcal{A}$ (referred to as adversarially-chosen-key attack).}
    
    \item \textit{{\rm $\textbf{Forge.}$} $\mathcal{A}$ outputs a signature {\rm $(\mu^*,\textsf{R}^*,\Sigma^*)$} and succeeds if (1) {\rm $\textsf{Ver}(\mu^*,\textsf{R}^*,\Sigma^*)=1$}, (2) {\rm $\textsf{R}^*\subseteq \textsf{S}$}, and (3) {\rm $(\mu^*,\textsf{R}^*,$ $\Sigma^*)\notin \textsf{L}$}.}
    \end{itemize}

\noindent$\textbf{Anonymity}$. \textit{A {\rm $\textsf{RS}$} scheme is signer-unclaimability, if for any PPT adversary $\mathcal{A}$, it holds that $\mathcal{A}$ has at most negligible advantage in the following experiment with a challenger $\mathcal{C}$.}
\begin{itemize}
    \item $\textbf{Setup}$. \textit{$\mathcal{C}$ generates {\rm $\textsf{PP}\leftarrow \textsf{Setup}(1^n;\gamma_{\textsf{st}})$} and {\rm $(\textsf{vk}^{(i)},\textsf{sk}^{(i)})\leftarrow \textsf{KeyGen}(\gamma_{\textsf{kg}}^{(i)})$} for all $i\in[N]$, where $N=\mathrm{poly}(n)$ and {\rm $(\gamma_{\textsf{st}}, \{\gamma_{\textsf{kg}}^{(i)}\}_{i\in [N]})$} are randomness used in {\rm $\textsf{Setup}$} and {\rm $\textsf{KeyGen}$}, respectively. $\mathcal{C}$ sets {\rm $\textsf{S}=\{\textsf{vk}^{(i)}\}_{i\in[N]}$}. Finally, $\mathcal{C}$ sends {\rm $(\textsf{PP}$, $\textsf{S},\gamma_{\textsf{st}},\{\gamma_{\textsf{kg}}^{(i)}\}_{i\in [N]})$} to $\mathcal{A}$.}
    
    \item $\textbf{Challenge}$. \textit{$\mathcal{A}$ provides a challenge {\rm $(\mu^*,\textsf{R}^*,s_0^*,s_1^*)$} to the challenger such that $s_0^*, s_1^*\in[N]$, $s_0^*\neq s_1^*$ and {\rm $\textsf{vk}^{(s_0^*)},\textsf{vk}^{(s_1^*)} \in \textsf{S}\bigcap \textsf{R}^*$}. $\mathcal{C}$ chooses a random bit $b\in\{0,1\}$ and computes the signature $\Sigma^*$ by invoking {\rm $\Sigma^*\leftarrow \textsf{Sign}(\mu^*,\textsf{R}^*,\textsf{sk}^{(s_b^*)}; \gamma_{\textsf{sign}})$}. Finally, returns {\rm$(\Sigma^*,\gamma_{\textsf{sign}})$} to $\mathcal{A}$.}
       
    \textsf{Note that we not only give $\mathcal{A}$ the randomness $\gamma_{\textsf{st}}$ and $\{\gamma_{\textsf{kg}}^{(i)}\}_{i\in [N]}$ in $\textsf{Setup}$ phase, but also give $\mathcal{A}$ the signature $\Sigma^*$ and the corresponding signing randomness $\gamma_{\textsf{sign}}$ that used to produce $\Sigma^*$ in $\textsf{Challenge}$ phase (referred to as signer-unclaimability).}
      
    \item $\textbf{Guess}$. \textit{$\mathcal{A}$ outputs a guess $b'$. If $b'=b$, $\mathcal{C}$ outputs 1, otherwise 0.}
    
\end{itemize}

\section{Preliminary}\label{Preliminary}

In this section, we first review some lattice-based backgrounds, then we review the key-homomorphic evaluation algorithm which we will use as a building block for our construction.

\noindent$\textbf{Notation}$. We denote vectors as lower-case bold letters (e.g. $\bold{x}$), and matrices by upper-case bold letters (e.g. $\bold A$). We say that a function in $n$ is $\textit{negligible}$, written negl$(n)$, if it vanishes faster than the inverse of any polynomial in $n$. We say that a probability $p(n)$ is $\textit{overwhelming}$ if $1-p(n)$ is negligible. We denote the horizontal concatenation of two matrices $\bold A$ and $\bold B$ as $\bold A|\bold B$.

\noindent$\textbf{Matrix Norms}$. For a vector $\bold x$, we let $\lVert \bold x \rVert$ denote its $l_2$-norm. For a matrix $\bold A$ we denote two matrix norms: $\lVert \bold A \rVert$ denotes the $l_2$ length of the longest column of $\bold A$. $\lVert \bold {\widetilde{A}} \rVert$ denotes the result of applying Gram-Schmidt orthogonalization to the columns of $\bold A$.

\noindent$\textbf{Lattices and Gaussian Distributions}$. Let $m\in\mathbb{Z}$ be a positive integer and $\bold{\Lambda}\subset \mathbb{R}^m$ be an $m$-dimensional full-rank lattice formed by the set of all integral combinations of $m$ linearly independent basis vectors $\bold B=(\bold b_1,\dots,\bold b_m)\subset \mathbb{Z}^{m}$, i.e., $\bold{\Lambda}=\mathcal{L}(\bold B)=\big\{\bold B\bold c=\sum\nolimits_{i=1}^m c_i\bold b_i:\bold c\in\mathbb{Z}^m\big\}$. For positive integers $n$, $m$, $q$, a matrix $\bold A\in \mathbb{Z}_q^{n\times m}$, and a vector $\bold y \in \mathbb{Z}_q^m$, the $m$-dimensional integer lattice $\bold{\bold{\Lambda}}_q^{\perp}(\bold A)$ is defined as $\bold{\bold{\Lambda}}_q^{\perp}(\bold A)=\{\bold x\in \mathbb{Z}^m:\bold{Ax}=\bold 0\pmod q\}$. $\bold{\bold{\Lambda}}_q^{\bold y}(\bold A)$ is defined as $\bold{\bold{\Lambda}}_q^{\bold y}(\bold A)=\{\bold x\in \mathbb{Z}^m:\bold{Ax}=\bold y\pmod q\}$. For a vector $\bold c\in\mathbb{R}^m$ and a positive parameter $\sigma\in\mathbb{R}$, define $\rho_{\sigma,\bold c}(\bold x)={\rm exp}(-\pi\lVert \bold x-\bold c \rVert/\sigma^2)$ and $\rho_{\sigma,\bold c}(\bold{\Lambda})=\sum\nolimits_{\bold x\in\bold{\Lambda}} \rho_{\sigma,\bold c}(\bold x)$. For any $\bold y\in \bold{\Lambda}$, define the discrete Gaussian distribution over $\bold{\Lambda}$ with center $\bold c$ and parameter $\sigma$ as $\mathcal{D}_{\bold{\Lambda},\sigma,\bold c}(\bold y)=\rho_{\sigma,\bold c}(\bold y)/\rho_{\sigma,\bold c}( \bold{\Lambda})$. For simplicity, $\rho_{\sigma,\bold 0}$ and $\mathcal{D}_{\bold{\Lambda},\sigma,\bold 0}$ are abbreviated as $\rho_{\sigma}$ and $\mathcal{D}_{\bold{\Lambda},\sigma}$, respectively.

The following Lemma \ref{Lemma1} bounds the length of a discrete Gaussian vector with a sufficiently large Gaussian parameter.

\noindent \begin{lemma}[{\rm \cite{MR07}}]
\noindent For any lattice $\bold{\Lambda}$ of integer dimension $m$ with basis $\bold B$, $\bold c\in \mathbb{R}^m$ and Gaussian parameter $\sigma>\lVert \widetilde{\bold{B}} \rVert \cdot \omega(\sqrt{{\rm log}\ m})$, we have $\mathrm{Pr}[\lVert \bold x-\bold c\rVert > \sigma\sqrt m : \bold x\leftarrow \mathcal{D}_{\bold{\Lambda},\sigma,\bold c}]\leq \mathrm{negl}(n)$.
\label{Lemma1} 
\end{lemma}

The following generalization of the leftover hash lemma is needed for our security proof.

\begin{lemma}[{\rm \cite{ABB10a}}]
Suppose that $m>(n+1){\rm log}~q+\omega({\rm log}~n)$ and that $q>2$ is prime. Let $\bold R$ be an $m\times k$ matrix chosen uniformly in $\{1,-1\}^{m\times k}~ {\rm mod}~q$ where $k=k(n)$ is polynomial in $n$. Let $\bold A$ and $\bold B$ be matrices chosen uniformly in $\mathbb{Z}_q^{n\times m}$ and $\mathbb{Z}_q^{n\times k}$, respectively. Then, for all vectors $\bold w$ in $\mathbb{Z}_q^{m}$, the distribution $(\bold A, \bold{AR}, \bold{R}^\top\bold w)$ is statistically close to the distribution $(\bold A, \bold{B}, \bold{R}^\top\bold w)$.
\label{LeftOverHash} 
\end{lemma}

The security of our \textsf{RS} construction is based on the following Small Integer Solution (\textsf{SIS}) assumption and the security of PRF.  

%\begin{definition}[SIS Assumption {\rm \cite{GPV08,MR07}}]
%Let $q,m,\beta$ be functions of $n$. Define {\rm $\textsf{SIS}_{q,n,m,\beta}$} problem as: Given a matrix $\bold A\xleftarrow{\$} \mathbb{Z}_q^{n\times m}$, find a non-zero vector $\bold x \in \mathbb{Z}^m$ s.t. $\bold{Ax}=\bold 0 \pmod q$ and $\lVert \bold x \rVert \leq \beta$. \\
%$~~~~~$For $m,\beta=\mathrm{poly}(n)$, $q\geq \beta\cdot \omega(\sqrt{n~{\rm log}~n})$, no (quantum) algorithm can solve {\rm $\textsf{SIS}_{q,n,m,\beta}$} problem in polynomial time.
%\label{DefofSIS}
%\end{definition}

\begin{definition}[SIS Assumption {\rm \cite{GPV08,MR07}}]
Let $q$ and $\beta$ be functions of $n$. An instance of the {\rm $\textsf{SIS}_{q,\beta}$} problem is a uniformly random matrix $\bold A\xleftarrow{\$} \mathbb{Z}_q^{n\times m}$ for any desired $m = \mathrm{poly}(n)$. The goal is to find a nonzero integer vector $\bold x\in \mathbb{Z}^m$ such that $\bold{Ax}=\bold 0 \pmod q$ and $\lVert \bold x \rVert \leq \beta$. \\
$~~~~~$For $\beta=\mathrm{poly}(n)$, $q\geq \beta\cdot \omega(\sqrt{n~{\rm log}~n})$, no (quantum) algorithm can solve {\rm $\textsf{SIS}_{q,\beta}$} problem in polynomial time.
\label{DefofSIS}
\end{definition}

\begin{definition}[Pseudorandom Functions]\label{DefofPRF}
For a security parameter $n>0$, let $k=k(n)$, $t=t(n)$ and $c=c(n)$. A pseudorandom function {\rm $\textsf{PRF}:\{0,1\}^k\times \{0,1\}^t\rightarrow \{0,1\}^c$} is an efficiently computable, deterministic two-input function where the first input, denoted by $K$, is the key. Let $\Omega$ be the set of all functions that map $\ell$ bits strings to $c$ bits strings. There is a negligible function {\rm $\textsf{negl}(n)$} such that:
\begin{center}
    {\rm $\Big| \mathrm{Pr}\big[\mathcal{A}^{\textsf{PRF}(K,\cdot)}(1^n)=1\big]-\mathrm{Pr}\big[\mathcal{A}^{F(\cdot)}(1^n)=1\big]\Big| \leq \textsf{negl}(n)$}
\end{center}
where the probability is taken over a uniform choice of key $K\xleftarrow{\$}\{0,1\}^k$, $F\xleftarrow{\$} \Omega$, and the randomness of $\mathcal{A}$. 
\end{definition}

\noindent \textit{Remark:} The PRF that we employed is Lai et al.'s work \cite{LLW20}, which is based on standard lattice assumption i.e., learning with errors (LWE) assumption.

\noindent \textbf{Algorithms on Lattices.} Our work will use the following lattice algorithms.

\begin{lemma}[{\rm \textsf{TrapGen}} Algorithm {\rm \cite{AP11}}]
Let $n\geq 1, q\geq 2$ and $m=O(n~{\rm log}~q)$ be integers. There is a probabilistic algorithm {\rm $\textsf{TrapGen}(1^n,1^m,q)$} that outputs a matrix $\bold A\in \mathbb{Z}_q^{n\times m}$ and a trapdoor $\bold {S_A}\subset \bold{\bold{\Lambda}}_q^{\perp}(\bold A)$, the distribution of $\bold A$ is statistically close to the uniform distribution over $\mathbb{Z}_q^{n\times m}$ has $\lVert \bold{\widetilde{\bold{S_A}}} \rVert\leq O(\sqrt{n\ {\rm log}\ q})$ and $\lVert \bold{S_A} \rVert\leq O(n~{\rm log}~q)$ with all but negligible probability in $n$. 
\label{TrapGen}
\end{lemma}

\begin{lemma}[{\rm \textsf{BasisExt}} Algorithm {\rm \cite{CHKP10}}]
For $i=1,2,3$, let $\bold A_i$ be a matrix in $\mathbb{Z}_q^{n\times m_i}$ whose columns generate $\mathbb{Z}_q^n$ and let $\bold A'=[\bold A_1|\bold A_2|\bold A_3]$. Let $\bold {S}_{\bold A_2}$ be a basis of $\bold{\Lambda}^\perp(\bold {S}_{\bold A_2})$. There is a deterministic algorithm {\rm $\textsf{BasisExt}(\bold A',\bold {S}_{\bold A_2})$} that outputs a basis $\bold {S}_{\bold A'}$ for ${\rm \bold{\Lambda}}^\perp(\bold A')$ such that $\lVert \widetilde{\bold {S_{A'}}}\rVert=\lVert \widetilde{\bold {S}_{\bold A_2}}\rVert$.
\label{BasisExt}
\end{lemma}

\begin{lemma}[{\rm \textsf{BasisRand}} Algorithm {\rm \cite{CHKP10}}]
Let $\bold {S_{A'}}\in \mathbb{Z}^{m'\times m'}$ be an extended basis of ${\rm \bold{\Lambda}}^\perp(\bold A')$ output by {\rm \textsf{BasisExt}}. There is a probabilistic algorithm {\rm $\textsf{BasisRand}(\bold {S_{A'}},\sigma)$} which takes as input a basis $\bold {S_{A'}}$ and a parameter $\sigma\geq \lVert \widetilde{\bold {S_{A'}}}\rVert\cdot \omega(\sqrt{{\rm log}~m})$, outputs a basis $\bold {S_{A''}}\in \mathbb{Z}^{m'\times m'}$ of ${\rm \bold{\Lambda}}^\perp(\bold A')$ which is statistically independent with the original basis $\bold {S_{A'}}$, and has $\lVert \widetilde{\bold {S_{A''}}}\rVert\leq \sigma\cdot \sqrt{m'}$ holds. 
\label{BasisRand}
\end{lemma}

The following lemma is a property of the \textsf{BasisRand} algorithm, which will be used in our signer-anonymity proof.

\begin{lemma}[Trapdoor Indistinguishability of {\rm \textsf{BasisRand}} {\rm \cite{CHKP10}}]
For any two basis $\bold S_0,\bold S_1$ of the same lattice and any $\sigma\geq \mathrm{max}\big\{\lVert \widetilde{\bold S_0}\rVert,\lVert \widetilde{\bold S_1}\rVert\big\}\cdot \omega(\sqrt{{\rm log}~m})$, the outputs of {\rm $\textsf{BasisRand}(\bold S_0,\sigma)$} and {\rm $\textsf{BasisRand}(\bold S_1,\sigma)$} are within $\mathrm{negl}(n)$ statistical distance.
\label{PropertyOfBasisRand}
\end{lemma}

The following lattice basis extension algorithm is also needed for our security proof, which was presented by Agrawal, Boneh, and Boyen \cite{ABB10a}, so we abbreviate that as \textsf{BasisExtABB} algorithm.

\begin{lemma}[{\rm \textsf{BasisExtABB}} Algorithm {\rm \cite{ABB10a}}]
Let $q$ be a prime, $n,m$ be integers with $m>n$. There is a probabilistic algorithm {\rm $\textsf{BasisExtABB}(\bold A,\bold {B},\bold R, \bold {S_B})$} which takes as input two matrices $\bold A,\bold B\in \mathbb{Z}_q^{n\times m}$ whose columns generate $\mathbb{Z}_q^n$, a matrix $\bold R\in\mathbb{Z}^{m\times m}$, and a basis $\bold {S_B}\in\bold{\Lambda}_q^\perp(\bold B)$, outputs a basis $\bold {S_F}$ of $\bold{\Lambda}_q^\perp(\bold F)$ such that $\lVert \widetilde{\bold {S_F}}\rVert<(\lVert {\bold {R}}\rVert+1)\cdot\lVert \widetilde{\bold {S_B}}\rVert$ where $\bold F=[\bold A| \bold {AR}+\bold B]\in \mathbb{Z}_q^{n\times 2m}$.
\label{BasisExtABB}
\end{lemma}

\begin{lemma}[{\rm \textsf{SampleGaussian}} Algorithm {\rm \cite{GPV08}}]
Let $q>2$, $m>n$ be integers. There is a probabilistic algorithm {\rm $\textsf{SampleGaussian}(\bold A,\bold {S_A},\bold y,\sigma)$} which takes as input a matrix $\bold A\in \mathbb{Z}_q^{n\times m}$ whose columns generate $\mathbb{Z}_q^n$, and a basis $\bold {S_A}$ of $\bold{\Lambda}_q^\perp(\bold A)$, a vector $\bold y\in \mathbb{Z}_q^{n}$, and a Gaussian parameter $\sigma\geq\lVert \widetilde{\bold{S_A}}\rVert \cdot \omega(\sqrt{{\rm log}~m})$, outputs a vector $\bold x\in \bold{\Lambda}_q^{\bold y}(\bold A)$ sampled from a distribution which is statistically close to $\mathcal D_{\bold{\Lambda}_q^{\bold y}(\bold A),\sigma}$.
\label{SampleGaussian}
\end{lemma}

Given output values of the algorithms \textsf{SampleGaussian}, the following algorithm \textsf{ExplainGaussian} is used to sample the randomness under which the \textsf{SampleGaussian} algorithm produces the desired output. 

\begin{lemma}[{\rm \textsf{ExplainGaussian}} Algorithm \cite{PS19}]
There is a probabilistic algorithm {\rm \textsf{ExplainGaussian}}$(\bold A,\bold {S_A}, \bold x, \sigma, \bold y)$ that on input a pair of matrices $(\bold A,\bold {S_A})$ from {\rm\textsf{TrapGen}}, preimage vector $\bold x$, a parameter $\sigma$ and a image vector $\bold y\in \mathbb{Z}^{m}$, samples randomness $\gamma$ that yields output $\bold x$ under algorithm {\rm \textsf{SampleGaussian}}, i.e., samples from the distribution $\{\gamma |$ {\rm \textsf{SampleGaussian}}$(\bold A,\bold {S_A},\bold y, \sigma; \gamma)=\bold x\}$.
\label{ExplainGaussian}
\end{lemma}

The following lemma is a property of the \textsf{ExplainGaussian} algorithm, which will be used in our signer-anonymity proof.

\begin{lemma}[Randomness Indistinguishability of {\rm \textsf{ExplainGaussian}} {\rm \cite{PS19}}]
Let $\mathcal{R}$ be the randomness space, let $(\bold A_0,\bold {S}_{\bold{A}_0})$ and $(\bold A_1,\bold {S}_{\bold{A}_1})$ be two pairs of matrices from {\rm\textsf{TrapGen}}, let $\bold F = [\bold A_0 | \bold A_1]$, let $\bold {S}_{\bold{F}_0}$ and $\bold {S}_{\bold{F}_1}$ be the extended basis from $\bold {S}_{\bold{A}_0}$ and $\bold {S}_{\bold{A}_1}$ respectively. The distribution of randomness $\gamma^{(0)}\xleftarrow{\$} \mathcal{R}$ for {\rm $\bold x \leftarrow \textsf{SampleGaussian}(\bold F, \bold {S}_{\bold{F}_0}, \bold y, \sigma; \gamma^{(0)})$} and {\rm $\gamma^{(1)}\leftarrow\textsf{ExplainGaussian}(\bold F, \bold {S}_{\bold{F}_1}, \bold x, \sigma, \bold y)$} are statistically indistinguishable.
\label{PropertyOfExplainGaussian}
\end{lemma}

\noindent \textbf{Gadget Matrix.} The ``gadget matrix'' $\bold G$ defined in the work \cite{MP12}. We recall the following one fact.
\begin{lemma}[{\rm \cite{MP12}}]
Let $q$ be a prime, and $n$, $m$ be integers with $m=n~{\rm log}~q$. There is a fixed full-rank matrix such that the lattice $\bold{\Lambda}_q^\perp(\bold G)$ has a publicly known basis $\bold{S_G}\in \mathbb{Z}^{m\times m}$ with $\lVert \widetilde{\bold{S_G}}\rVert\leq \sqrt{5}$.
\label{GTrap}
\end{lemma}

\subsection{Key-Homomorphic Evaluation Algorithm}\label{KeyHomo}

In our construction, we borrow the idea from the standard signature work \cite{BL16}, that is employing the key-homomorphic evaluation algorithm $\textsf{Eval}(\cdot, \cdot)$ from \cite{GSW13,BV14,BGG+14} to evaluate circuits of a PRF. In particular, they used the Brakerski and Vaikuntanathan’s evaluation algorithm \cite{BV14}. The inputs of $\textsf{Eval}(\cdot, \cdot)$ are ${C}$ and a set of $\ell$ different matrices $\{\bold A^{(i)}\}_{i\in[\ell]}$, where ${C}: \{0,1\}^\ell\rightarrow \{0,1\}$ is a fan-in-2 Boolean $\textsf{NAND}$ circuit expression of some functions such as a PRF, and each $\bold A^{(i)}=\bold A\bold R^{(i)}+b^{(i)}\bold G\in\mathbb{Z}_q^{n\times m}$ corresponds to each input wire of ${C}$, and where $\bold A\xleftarrow{\$} \mathbb{Z}_q^{n\times m}$, $\bold R^{(i)}\xleftarrow{\$} \{1,-1\}^{m\times m}$, $b^{(i)}\in\{0,1\}$ and $\bold G\in\mathbb{Z}_q^{n\times m}$ is the gadget matrix. The algorithm deterministically output a matrix $\bold A_{C}=\bold A\bold R_{C}+{C}\big(b^{(1)},\dots,b^{(\ell)}\big)\bold G\in\mathbb{Z}_q^{n\times m}$. In the analyzation of our unforgeability proof, we will use the following lemma to show $\bold R_{C}$ is short enough.

\begin{lemma}
Let $C:\{0,1\}^\ell \rightarrow \{0,1\}$ be a {\rm $\textsf{NAND}$} boolean circuit which has depth $d=c~\mathrm{log}~\ell$ for some constant $c$. Let $\big\{\bold A^{(i)}=\bold A\bold R^{(i)}+b^{(i)}\bold G\in\mathbb{Z}_q^{n\times m}\big\}_{i\in[\ell]}$ be $\ell$ different matrices correspond to each input wire of ${C}$ where $\bold A\xleftarrow{\$} \mathbb{Z}_q^{n\times m}$, $\bold R^{(i)}\xleftarrow{\$} \{1,-1\}^{m\times m}$, $b^{(i)}\in\{0,1\}$ and $\bold G\in\mathbb{Z}_q^{n\times m}$ is the gadget matrix. There is an efficient deterministic evaluation algorithm {\rm $\textsf{Eval}\big(C, (\bold A^{(1)},\cdots, \bold A^{(\ell)})\big)$} runs in time $\mathrm{poly}(4^d,\ell,n,\mathrm{log}~q)$, the inputs are $C$ and $\{\bold A^{(i)}\}_{i\in[\ell]}$, the output is a matrix 
{\rm\begin{align*}
 \bold A_C&=\bold A\bold R_{C}+{C}\big(b^{(1)},\dots,b^{(\ell)}\big)\bold G\\
 		&=\textsf{Eval}\big(C, (\bold A^{(1)},\dots,\bold A^{(\ell)})\big)
\end{align*}}
 where ${C}\big(b^{(1)},\dots,b^{(\ell)}\big)$ is the output bit of $C$ on the arguments $\big(b^{(1)},\dots,b^{(\ell)}\big)$ and $\bold R_{C}\in \mathbb{Z}^{m\times m}$ is a low norm matrix has $\lVert\bold R_C\rVert \leq O(\ell^{2c}\cdot m^{3/2})$. 
\label{Lemma7}
\end{lemma}

\section{Our Scheme}\label{OurConstruction}

In this section, we present the construction of our $\textsf{RS}$ scheme in Sect. \ref{Construction}, and give the concrete parameters in Sect. \ref{Correctness}. Then we prove the unforgeability and anonymity in Sect. \ref{Unforgeability} and Sect. \ref{Anonymity}, respectively.

\subsection{Construction}\label{Construction}

\noindent$\textsf{Setup}(1^n;\gamma_{\textsf{st}})$
\vspace{-0.5em}
\begin{enumerate}
	\item On input a security parameter $n$, sets the modulo $q$, lattice dimension $m$, PRF key length $k$, message length $t$, let $\gamma_{\textsf{st}}$ be the randomness that use to choose Gaussian parameters and then chooses Gaussian parameters $\sigma$ and $\sigma'$ as specified in Sect. \ref{Correctness} below.	
	\item Select a secure $\textsf{PRF}:\{0,1\}^k\times \{0,1\}^t \rightarrow \{0,1\}$, express it as a $\textsf{NAND}$ Boolean circuit ${C_\textsf{PRF}}$.
	\item Output $\textsf{PP}=(q,m,k,t,\sigma,\sigma', \textsf{PRF},\gamma_{\textsf{st}})$.
\end{enumerate}
Note that including the randomness $\gamma_{\textsf{st}}$ in $\textsf{PP}$ is to guarantee the public has no concerns about the existence of trapdoors.

In the following, $\textsf{PP}$ are implicit input parameters to every algorithm.
\vspace{1em}

\noindent$\textsf{KeyGen}()$
\vspace{-0.5em}
\begin{enumerate}
	
	\item Sample $(\bold A,\bold S_{\bold A})\leftarrow \textsf{TrapGen}(1^n,1^m,q)$ where $\bold A\in \mathbb{Z}_q^{n\times m}$, $\bold S_{\bold A}\in \mathbb{Z}^{m\times m}$.
	\item Select a $\textsf{PRF}$ key $\bold k=\big(k_1,k_2,\dots,k_k\big)\xleftarrow{\$}\{0,1\}^k$.

	\item Select $\bold A_0, \bold A_1, \bold C_0, \bold C_1 \xleftarrow{\$}\mathbb{Z}_q^{n\times m}$.

	\item For $j=1$ to $k$, select $\bold B_{j}\xleftarrow{\$} \mathbb{Z}_q^{n\times m}$.
	
	\item Output $\textsf{vk}=(\bold A,(\bold A_{0},\bold A_{1}), \{\bold B_{j}\}_{j\in[k]},(\bold C_{0},\bold C_{1}))$ and $\textsf{sk}=(\bold S_{\bold A}, \bold k)$.
\end{enumerate}
\vspace{0.5em}

\noindent$\textsf{Sign}(\boldsymbol{\mu},\textsf{R},\textsf{sk})$
\vspace{-0.2em}
\begin{enumerate}
	\item On input a message $\boldsymbol{\mu}=(\mu_1,\dots,\mu_t)\in\{0,1\}^t$, a ring of verification keys $\textsf{R}=(\textsf{vk}^{(1)},\dots,\textsf{vk}^{(N)})$ where \[\textsf{vk}^{(i)}=(\bold A^{(i)},(\bold A_{0}^{(i)},\bold A_{1}^{(i)}), \{\bold B_{j}^{(i)}\}_{j\in[k]},(\bold C_{0}^{(i)},\bold C_{1}^{(i)}))\] and a signer's signing key $\textsf{sk}:=\textsf{sk}^{(\bar i)}$ where $\bar i \in [N]$ be the index of the signer in the ring $\textsf{R}$.

	\item Compute $b=\textsf{PRF}(\bold k^{\bar i}, \boldsymbol{\mu})$.

	\item For $i=1$ to $N$, compute $\bold A_{C_\textsf{PRF},\boldsymbol{\mu}}^{(i)}\in \mathbb{Z}_q^{n\times m}$ by \[\bold A_{C_\textsf{PRF},\boldsymbol{\mu}}^{(i)}=\textsf{Eval}\big(C_\textsf{PRF},(\{\bold B_{j}^{(i)}\}_{j\in[k]},\bold C_{\mu_1}^{(i)},\dots,\bold C_{\mu_t}^{(i)})\big)\] and set $\bold F_{C_\textsf{PRF},\boldsymbol{\mu},1-b}^{(i)}=\big[\bold A^{(i)} | \bold A_{1-b}^{(i)}-\bold A_{C_\textsf{PRF},\boldsymbol{\mu}}^{(i)}\big]$.

	\item Let $\bold F'_{1-b}=\big[\bold F_{C_\textsf{PRF},\boldsymbol{\mu},1-b}^{(1)}|\dots|\bold F_{C_\textsf{PRF},\boldsymbol{\mu},1-b}^{(N)}\big]$. Compute \[\bold{S}_{\bold F'_{1-b}}\leftarrow\textsf{BasisRand}\big(\textsf{BasisExt}(\bold F'_{1-b}, \bold{S}_{\bold A^{(\bar i)}}), \sigma\big).\]

	\item Compute $\bold x'\leftarrow\textsf{SampleGaussian}(\bold F'_{1-b},\bold S_{\bold F'_{1-b}},\bold 0,\sigma')$ such that $\bold F'_{1-b}\cdot \bold x' = \bold 0 \pmod q$.
	
	\item Output the signature $\Sigma=\bold x'$.
\end{enumerate}
\vspace{0.5em}

\noindent$\textsf{Ver}(\boldsymbol{\mu},\textsf{R},\Sigma)$
\vspace{-0.2em}
\begin{enumerate}
	\item On input a message $\boldsymbol{\mu}=(\mu_1,\dots,\mu_t)\in\{0,1\}^t$, a ring of verification keys $\textsf{R}=(\textsf{vk}^{(1)},\dots,\textsf{vk}^{(N)})$, and a signature $\Sigma=\bold x'$.

	\item For $i=1$ to $N$, check if $\lVert{\bold x'}\rVert\leq \sigma\sqrt{2Nm}$ holds, otherwise return 0.

	\item For $i=1$ to $N$, compute $\bold A_{C_\textsf{PRF},\boldsymbol{\mu}}^{(i)}$ as in \textsf{Sign} algorithm. 
	\item For $b\in\{0,1\}$, set $\bold F_{C_\textsf{PRF},\boldsymbol{\mu},b}^{(i)}=\big[\bold A^{(i)} | \bold A_{b}^{(i)}-\bold A_{C_\textsf{PRF},\boldsymbol{\mu}}^{(i)}\big]$ and $\bold F'_b=\big[\bold F_{C_\textsf{PRF},\boldsymbol{\mu},b}^{(1)} |\dots|\bold F_{C_\textsf{PRF},\boldsymbol{\mu},b}^{(N)}\big]$. 
	\item Check if $\bold F'_b\cdot\bold x'=\bold 0\pmod q$ holds for $b=0$ or $1$, return 1, otherwise return 0.
\end{enumerate}

\subsection{Correctness and Parameters}\label{Correctness}
We now show the correctness of \textsf{RS}. By Lemma \ref{SampleGaussian} the signature $\Sigma=\bold x'$ follows the distribution $\mathcal{D}_{\bold{\Lambda}_q^\perp(\bold F'_b),\sigma'}$. By Lemma \ref{Lemma1}, the length of $\bold x'$ at most $\sigma\sqrt{2Nm}$ with overwhelming probability. Therefore, the signature is accepted by the \textsf{Ver} algorithm.

We then explain the parameters choosing. We employ the work \cite{LLW20} to instantiate our PRF, which based on standard LWE assumption with polynomial modulus $q=n^{\omega(1)}$. Let $n$ be the security parameter, let the message length be $t=t(n)$ and the secret key length of PRF be $k=k(n)$. Let $\ell=t+k$ be the input length of PRF. To ensure that hard lattices with good short bases can be generated by $\textsf{TrapGen}$ in Lemma \ref{TrapGen}, we need to set $m=6n^{1+\delta}$ where $\delta>0$ is a constant such that $n^{\delta}>O({\rm log}~n)$. To ensure the randomized basis is statistically independent with the original basis as required in Lemma \ref{BasisRand}, we need to set $\sigma=O\big(\ell^{2c}\cdot m^{3/2}\big)\cdot \omega\big(\sqrt{{\rm log}~Nm}\big)$ (see the unforgeability proof below). To ensure that the distribution on the output of $\textsf{SampleGaussian}$ statistically close to the distribution $\mathcal{D}_{\bold{\Lambda}_q^\perp(\bold F'),\sigma'}$, we need to set $\sigma'$ sufficiently large that is $\sigma'=\sqrt{N}\cdot O\big(\ell^{2c}\cdot m^{2}\big)\cdot \omega\big({{\rm log}~Nm}\big)$ (see the unforgeability proof below). To ensure that vectors sampled using a trapdoor are difficult SIS solutions, we need to set $\beta \geq O(\ell^{2c}\cdot m^{3/2})\cdot \sigma\sqrt{2m}$ for some constant $c$ (see the unforgeability proof below). To ensure our construction based on SIS has a worst-case lattice reduction as defined in Definition \ref{DefofSIS}, we need to set $q\geq \beta \cdot \omega(\sqrt{n~{\rm log}~n})$.

To satisfy the above requirements, let $n$ be the security parameter, the other parameters can be instantiated in various ways. For a typical choice, we choose a function $\omega(\sqrt{{\rm log}~m})$, $N=N(n)$, set the parameters $(m,\sigma,\sigma',\beta,q)$ as follows
\begin{gather*}
    m=6n^{1+\delta}\\ 
    \sigma=O\big(\ell^{2c}\cdot m^{3/2}\big)\cdot \omega\big(\sqrt{{\rm log}~Nm}\big)\\
    \sigma'=\sqrt{N}\cdot O\big(\ell^{2c}\cdot m^{2}\big)\cdot \big(\omega\big(\sqrt{{\rm log}~Nm}\big)\big)^2\\
    \beta=N\cdot O(\ell^{4c}\cdot m^{7/2})\cdot \omega\big(\sqrt{{\rm log}~Nm}\big)\\
    q=N\cdot O(\ell^{4c}\cdot m^{4})\cdot \big(\omega\big(\sqrt{{\rm log}~Nm}\big)\big)^2
\end{gather*}

\subsection{Unforgeability}\label{Unforgeability}
We now prove the unforgeability of $\textsf{RS}$. 

\begin{theorem}[Unforgeability]\label{UnforgeabilityProof}
Let $m,q,\beta,\sigma$ be some polynomials in the security parameter $n$. For large enough $\sigma=O\big(\ell^{2c}\cdot m^{3/2}\big)\cdot \omega\big(\sqrt{{\rm log}~Nm}\big)$, $\sigma'=\sqrt{N}\cdot O\big(\ell^{2c}\cdot m^{2}\big)\cdot \big(\omega\big(\sqrt{{\rm log}~Nm}\big)\big)^2$ and $\beta \geq N\cdot O(\ell^{2c}\cdot m^{3/2})\cdot \sigma\sqrt{2m}$, if the hardness assumption {\rm $\textsf{SIS}_{q,\beta}$} holds and the based PRF is secure, the {\rm \textsf{RS}} scheme is {\rm \textsf{UnfAdvKey}} secure.
\end{theorem}

\noindent \textit{Proof.} Consider the following security game between a adversary $\mathcal{A}$ and a simulator $\mathcal{S}$. Upon receiving a challenge $\bold A\in \mathbb{Z}_q^{n\times m'}$ that is formed by $m'=m\cdot N$ uniformly random and independent samples from $\mathbb{Z}_q^{n}$, parsing $\bold A$ as $\bold A = \big[\bold A^{(1)}|\dots|\bold A^{(N)}\big]$, $\mathcal{S}$ simulates as follows.

\noindent$\textbf{Setup Phase}$. $\mathcal{S}$ takes as input a security parameter $n$ and a randomness $\gamma_{\textsf{st}}$ to invoke $\textsf{PP}\leftarrow \textsf{Setup}(1^n;\gamma_{\textsf{st}})$ algorithm. Then $\mathcal{S}$ simulates as follows:

\noindent Select a PRF key $\bold k=(k_1,k_2,\dots,k_k)\xleftarrow{\$}\{0,1\}^k$.

\noindent For $i=1$ to $N$, $b\in \{0,1\}$: 
\begin{itemize}
	\item Choose $\bold R_{\bold A_{b}}^{(i)}, \bold R_{\bold C_{b}}^{(i)}\xleftarrow{\$} \{1,-1\}^{m\times m}$.
	\item Construct $\bold A_{b}^{(i)}=\bold A^{(i)}\bold R_{\bold A_{b}^{(i)}}+b\bold G$ and $\bold C_{b}^{(i)}=\bold A^{(i)}\bold R_{\bold C_{b}^{(i)}}+b\bold G$ where $\bold G$ is the gadget matrix. 

    \end{itemize}
    For $j=1$ to $k$: 

    \begin{itemize}
        \item Choose $\bold R_{\bold B_{j}^{(i)}}\xleftarrow{\$} \{1,-1\}^{m\times m}$ and construct $\bold B_{j}^{(i)}=\bold A^{(i)}\bold R_{\bold B_{j}^{(i)}}+k_j\bold G$. 

\end{itemize}

\noindent$\mathcal{S}$ sets $\textsf{vk}^{(i)}=\big(\bold A^{(i)},(\bold A_{0}^{(i)},\bold A_{1}^{(i)}), \{\bold B_{j}^{(i)}\}_{j\in[k]},(\bold C_{0}^{(i)},\bold C_{1}^{(i)})\big)$ and $\textsf{S}=\{\textsf{vk}^{(i)}\}_{i\in[N]}$, then sends $(\textsf{PP},\textsf{S},\gamma_{\textsf{st}})$ to $\mathcal{A}$.

\noindent$\textbf{Probing Phase}$. $\mathcal{A}$ adaptively issues tuples for querying the signing oracle $\textsf{OSign}(\cdot,\cdot,\cdot)$. For simplicity, here consider only one tuple $(\boldsymbol{\mu},\textsf{R},s)$ where $s\in[N]$, and requires that $\textsf{vk}^{(s)}\in \textsf{S}\cap\textsf{R}$. Assume the ring $\textsf{R}=\big(\textsf{vk}^{(1)},\dots,$ $\textsf{vk}^{(N')}\big)$, parse the $\textsf{vk}^{(s)}=(\bold A^{(s)},(\bold A_{0}^{(s)},\bold A_{1}^{(s)}), \{\bold B_{j}^{(s)}\}_{j\in[k]},(\bold C_{0}^{(s)},\bold C_{1}^{(s)}))$ and let $N'=\lvert\textsf{R}\rvert$. $\mathcal{S}$ does the following to response the signature.

	Compute $b=\textsf{PRF}(\bold k, \boldsymbol{\mu})$.
	
	For $i'=1$ to $N'$, compute the evaluated matrix $\bold A_{C_\textsf{PRF},\boldsymbol{\mu}}^{(i')}$ by $\textsf{Eval}\big(C_\textsf{PRF},(\{\bold B_{j}^{(i')}\}_{j\in[k]},\bold C_{\mu_1}^{(i')},\dots,\bold C_{\mu_t}^{(i')})\big)$. Then set
\begin{align*}
		\bold F_{C_\textsf{PRF},\boldsymbol{\mu}, 1-b}^{(i')}&=\big[\bold A^{(i')}\big| \bold A_{1-b}^{(i')}-\bold A_{C_\textsf{PRF},\boldsymbol{\mu}}^{(i')}\big]\\
		&=\big[\bold A^{(i')}\big| \bold A^{(i')} \big(\bold R_{1-b}^{(i')}-\bold R_{C_\textsf{PRF},\boldsymbol{\mu}}^{(i')}\big)+(1-2b)\bold G\big]
\end{align*}

	Let $\bold F'_{1-b}=\big [\bold F_{C_\textsf{PRF},\boldsymbol{\mu}, 1-b}^{(1)}|\dots|\bold F_{C_\textsf{PRF},\boldsymbol{\mu}, 1-b}^{(N')}\big ]$ and $\bar{\bold R}^{(s)}=\bold R_{1-b}^{(s)}-\bold R_{C_\textsf{PRF},\boldsymbol{\mu}}^{(s)}$. Invoking 
	\begin{gather*}
	\bold S_{\bold F_{C_\textsf{PRF},\boldsymbol{\mu}, 1-b}^{(i')}}\leftarrow \textsf{BasisExtABB}\big(\bold A^{(s)}, \bold G, \bar{\bold R}^{(s)}, \bold {S_G}\big)\\
	\bar{\bold S}_{\bold F'_{1-b}}\leftarrow \textsf{BasisExt}\big(\bold S_{\bold F_{C_\textsf{PRF},\boldsymbol{\mu}, 1-b}^{(i')}}, \bold F'_{1-b}\big)\\
	\bold{S_{F'_{1-b}}}\leftarrow\textsf{BasisRand}\big(\bar{\bold S}_{\bold F'_{1-b}}, \sigma\big)
	\end{gather*}
then compute $\bold x'\leftarrow\textsf{SampleGaussian}\big(\bold F'_{1-b}, \bold{S_{F'_{1-b}}}, \bold 0, \sigma'\big)$ such that $\bold F'_{1-b}\cdot \bold x' = \bold 0 \pmod q$. 

$\mathcal{S}$ responses the signature $\Sigma=\bold x'$ for the query tuple $\big(\boldsymbol{\mu},\textsf{R},s\big)$ to $\mathcal{A}$ and adds $(\boldsymbol{\mu},\textsf{R},\Sigma)$ to a list {\rm \textsf{L}} which $\mathcal{S}$ initialized in prior.

\noindent $\textbf{Exploiting the forgery}$. $\mathcal{A}$ outputs a forgery signature tuple $(\boldsymbol{\mu}^*,\textsf{R}^*,\Sigma^*)$. Let $N^*=|\textsf{R}^*|$. Parse $\boldsymbol{\mu}^*=(\mu_1^*,\dots,\mu_t^*)$, $\textsf{R}^*=({\textsf{vk}}^{(1)},\dots,{\textsf{vk}}^{(N^*)})$, $\textsf{vk}^{(i^*)}=(\bold A^{(i^*)},(\bold A_{0}^{(i^*)},\bold A_{1}^{(i^*)}),\{\bold B_{j}^{(i^*)}\}_{j\in[k]},(\bold C_{0}^{(i^*)},\bold C_{1}^{(i^*)}))$ and $\Sigma^*=\bold {x}'^*$ where $\bold {x}'^*=\big(({\bold {x}^*}^{(1)})^\top|\dots|({\bold {x}^*}^{(N^*)})^\top\big)^\top$. Separate ${\bold {x}^*}^{(i^*)}$ into $\big(({{\bold {x}}_1^*}^{(i^*)})^\top|({{\bold {x}}_2^*}^{(i^*)})^\top\big)^\top$. $\mathcal{S}$ does the following to exploit the forgery.
\begin{itemize}
	\item Check if $(\boldsymbol{\mu}^*,\textsf{R}^*,\Sigma^*) \in \textsf{L}$ or $\lVert {\bold {x}'}^*\rVert> \sigma\sqrt{2N^*m}$, $\mathcal{S}$ aborts.

	\item For $i^*=1$ to $N^*$, compute the matrix $\bold A_{C_\textsf{PRF},\boldsymbol{\mu}^*}^{(i^*)}$ as in the probing phase above. Then, compute the matrix $\bold F_{C_\textsf{PRF},\boldsymbol{\mu}^*, 1-b^*}^{(i^*)}=\Big[\bold A^{(i^*)}\big|\bold A_{1-b^*}^{(i^*)}-\bold A_{C_\textsf{PRF},\boldsymbol{\mu}^*}^{(i^*)}\Big]$ (resp., $\bold F_{C_\textsf{PRF},\boldsymbol{\mu}^*, b^*}^{(i^*)}=\Big[\bold A^{(i^*)}\big|\bold A_{b^*}^{(i^*)}-\bold A_{C_\textsf{PRF},\boldsymbol{\mu}^*}^{(i^*)}\Big]$) and $\bold F'_{1-b^*}=\Big[\bold F_{C_\textsf{PRF},\boldsymbol{\mu},1-b^*}^{(1)}\big|\dots\big|\bold F_{C_\textsf{PRF},\boldsymbol{\mu},1-b^*}^{(N)}\Big]$ (resp., $\bold F'_{b^*}=\Big[\bold F_{C_\textsf{PRF},\boldsymbol{\mu},b}^{(1)}\big|\dots\big|\bold F_{C_\textsf{PRF},\boldsymbol{\mu},b^*}^{(N)}\Big]$).

	\item Check if $\bold F'_{1-b^*}\cdot \bold {x'^*}=\bold 0\pmod q$ holds, $\mathcal{S}$ aborts. Therefore, it holds that $\bold F'_{b^*}\cdot \bold {x'^*}=\sum_{i^*\in[N^*]} \Big[\bold A^{(i^*)}\big|\bold A^{(i^*)}\big(\bold R_{\bold A_{b^*}^{(i^*)}}-\bold R_{C_\textsf{PRF},\boldsymbol{\mu}^*}^{(i^*)}\big)\Big]\cdot \bold {x^*}^{(i^*)}=\bold 0\pmod q$.

	\end{itemize}
	 Therefore, we have $\sum_{i^*\in[N^*]} \bold A^{(i^*)}\cdot {\bar {\bold x}^*}{^{(i^*)}}=\bold 0\pmod q$ where ${\bar {\bold x}^*}{^{(i^*)}}=\Big({{\bold x}_1^*}{^{(i^*)}}+\big(\bold R_{\bold A_{b^*}}^{(i^*)}-\bold R_{C_\textsf{PRF},\boldsymbol{\mu}^*}^{(i^*)}\big)\cdot {{\bold x}_2^*}{^{(i^*)}}\Big)$ then we have $\sum_{i^*\in[N^*]} \bold A^{(i^*)}\cdot{\bar {\bold x}^*}{^{(i^*)}}=\bold 0\pmod q$. Let ${\bar {\bold x}^*}$ be the concatenation of $\{{\bar {\bold x}^*}{^{(i^*)}}\}$ i.e., ${\bar {\bold x}^*}=\Big(\big({\bar {\bold x}^*}{^{(1^*)}}\big)^\top\big|\dots\big|\big({\bar {\bold x}^*}{^{(N^*)}}\big)^\top\Big)^\top$. Note that $\{\bold A^{(i^*)}\}_{i^*\in[N^*]}$ is a subset of $\{\bold A^{(i)}\}_{i\in[N]}$, and we know $\bold A=\big[\bold A^{(1)}|\dots|\bold A^{(N)}\big]$. Therefore, by inserting zeros into ${\bar {\bold x}^*}$, $\mathcal{S}$ can obtain a nonzero ${\hat {\bold x}^*}$ such that $\bold A{\hat {\bold x}^*}=\bold 0\pmod q$. Therefore, $\mathcal{S}$ can output ${\hat {\bold x}^*}$ as a {\rm $\textsf{SIS}_{q,\beta}$} solution.

\newtheorem{claim}{Claim}

\begin{claim}\label{Claim1}
The set of verifications keys {\rm $\textsf{S}$} that simulated by {\rm $\mathcal{S}$} is statistically close to those in the real attack.
\end{claim}

\begin{proof}\renewcommand{\qedsymbol}{}
In the real scheme, the matrices $\{\bold A^{(i)}\}_{i\in[N]}$ generated by $\textsf{TrapGen}$. In the simulation, $\{\bold A^{(i)}\}_{i\in[N]}$ have uniform distribution as it comes from the $\textsf{SIS}$ challenger that are formed by $m'$ uniformly random and independent samples from $\mathbb{Z}_q^{n}$. By Lemma \ref{TrapGen}, $\{\bold A^{(i)}\}_{i\in[N]}$ generated in the simulation has right distribution except a negligibly statistical error. For the matrices $(\bold A_{0}^{(i)},\bold A_{1}^{(i)}), \{\bold B_{j}^{(i)}\}_{j\in[k]}$ and $(\bold C_{0}^{(i)},\bold C_{1}^{(i)})$ generated in the simulation have distribution that is statistically close to uniform distribution in $\mathbb{Z}_q^{n\times m}$ by Lemma \ref{LeftOverHash}. Therefore, the set of verifications keys $\textsf{S}$ given to $\mathcal{A}$ is statistically close to those in the real attack.
\end{proof}

\begin{claim}\label{Claim2}
The replies of the signing oracle {\rm \textsf{OSign}$(\cdot,\cdot,\cdot)$} simulated by $\mathcal{S}$ is statistically close to those in the real attack when set $\sigma=O\big(\ell^{2c}\cdot m^{3/2}\big)\cdot \omega\big(\sqrt{{\rm log}~Nm}\big)$ and $\sigma'=\sqrt{N}\cdot O\big(\ell^{2c}\cdot m^{2}\big)\cdot \big(\omega\big(\sqrt{{\rm log}~Nm}\big)\big)^2$.
\end{claim}
\begin{proof}\renewcommand{\qedsymbol}{}
By Lemma \ref{SampleGaussian}, for sufficient large Gaussian parameter $\sigma'$, the distribution of the $\bold x'$ generated in the simulation by $\textsf{SampleGaussian}$ is statistically close to the distribution of signatures (i.e., $\mathcal{D}_{\bold{\Lambda}_q^\perp(\bold F'),\sigma}$) generated in the real scheme. So we next analyze how to set the parameter $\sigma'$. In the $\textsf{Simulating Signing Oracle}$ phase, we constructed 
\[
\bold F_{C_\textsf{PRF},\boldsymbol{\mu}, 1-b}^{(i')}=\big[\bold A^{(i')}\big| \bold A^{(i')} \big(\bold R_{1-b}^{(i')}-\bold R_{C_\textsf{PRF},\boldsymbol{\mu}}^{(i')}\big)+(1-2b)\bold G\big]
\]

\noindent Let $\bar{\bold R}^{(i')}=\bold R_{1-b}^{(i')}-\bold R_{C_\textsf{PRF},\boldsymbol{\mu}}^{(i')}$. By Lemma \ref{Lemma7}, we know $\big\lVert\widetilde{ \bar{\bold R}^{(i')}} \big\rVert \leq O(\ell^{2c}\cdot m^{3/2})$ for some constant $c$. By Lemma \ref{BasisExtABB}, we know $\Big\lVert \widetilde{\bold S_{\bold F_{C_\textsf{PRF}}^{(i')}}}_{,\boldsymbol{\mu}, 1-b} \Big\rVert < \big(\big\lVert \bar{\bold R}^{(i')} +1\big\rVert\big) \cdot \big\lVert \widetilde{\bold {S_G}}\big\rVert$. Let $\bold F'_{1-b}=\Big[\bold F_{C_\textsf{PRF},\boldsymbol{\mu},1-b}^{(1)}\big|\dots\big|\bold F_{C_\textsf{PRF},\boldsymbol{\mu},1-b}^{(N)}\Big]\in \mathbb{Z}_q^{n\times 2Nm}$ be the extended basis output by \textsf{BasisExt}. By Lemma \ref{BasisRand}, it requires to set $\sigma>\big\lVert \widetilde{\bold S_{\bold F'_{1-b}}}\big\rVert \cdot \omega(\sqrt{{\rm log}~m})$. Let $\bold{F}''_{1-b}$ be the randomized basis output by \textsf{BasisRand}. By Lemma \ref{BasisRand}, we know $\lVert \widetilde{\bold S_{\bold F''_{1-b}}}\rVert \leq \sigma\sqrt{Nm}$. By Lemma \ref{SampleGaussian}, it requires to set $\sigma'>\lVert \widetilde{\bold S_{\bold F''_{1-b}}}\rVert \cdot \omega(\sqrt{{\rm log}~Nm})$. Therefore, to satisfy these requirements, set $\sigma=O\big(\ell^{2c}\cdot m^{3/2}\big)\cdot \omega\big(\sqrt{{\rm log}~Nm}\big)$ and $\sigma'=\sqrt{N}\cdot O\big(\ell^{2c}\cdot m^{2}\big)\cdot \big(\omega\big(\sqrt{{\rm log}~Nm}\big)\big)^2$ is sufficient.
\end{proof}

\begin{claim}\label{Claim4}
It's hard for $\mathcal{A}$ to find a messages $\boldsymbol{\mu}'$ such that each $\bold A_{C_\textsf{PRF},\boldsymbol{\mu}}^{(i)}=\bold A_{C_\textsf{PRF},\boldsymbol{\mu}'}^{(i)}$ holds.
\end{claim}
\begin{proof}\renewcommand{\qedsymbol}{}
In our construction, note that $\bold F_{C_\textsf{PRF},\boldsymbol{\mu}, 1-b}^{(i')}=\Big[\bold A^{(i')}\big| \bold A_{1-b}^{(i')}-\bold A_{C_\textsf{PRF},\boldsymbol{\mu}}^{(i')}\Big]$, one attacking method is to find a messages $\boldsymbol{\mu}'$ and a $\bar k'$ such that each $\bold A_{C_\textsf{PRF},\boldsymbol{\mu}}^{(i)}=\bold A_{C_\textsf{PRF},\boldsymbol{\mu}'}^{(i)}$ holds. Assume an efficient adversary can do that, with the public parameters constructed above, it holds that $\bold A^{(i')}\bold R_{C_\textsf{PRF},\boldsymbol{\mu}}^{(i')}+\textsf{PRF}\big(\bold k^{(i')}, \boldsymbol{\mu}\big)\bold G=\bold A^{(i')}\bold R_{C_\textsf{PRF},\boldsymbol{\mu}'}^{(i')} +\textsf{PRF}\big(\bold k^{(i')}, \boldsymbol{\mu}'\big)\bold G$.
Assume the based {PRF} is secure, with $1/2$ probability that $\textsf{PRF}\big(\bold k^{(i')}, \boldsymbol{\mu}\big) \neq \textsf{PRF}\big(\bold k^{(i')}, \boldsymbol{\mu}'\big)$ holds. In this case, we have $\bold R_{C_\textsf{PRF},\boldsymbol{\mu}}^{(i')}\neq \bold R_{C_\textsf{PRF},\boldsymbol{\mu}'}^{(i')}$ and $\bold A^{(i')}\big(\bold R_{C_\textsf{PRF},\boldsymbol{\mu}}^{(i')}-\bold R_{C_\textsf{PRF},\boldsymbol{\mu}'}^{(i')}\big)\pm \bold G=\bold 0 \pmod q$ holds. By Lemma \ref{SampleGaussian} and \ref{GTrap}, a low-norm vector $\bold e \in\mathbb{Z}^{m}$ can be efficiently found such that $\bold{Ge}=\bold 0\pmod q$ where $\bold e \neq \bold 0$ and $\lVert \bold e\rVert \leq \sigma_{\bold G} \sqrt{m}$ for some parameter $\sigma_{\bold G}\geq \sqrt{5}\cdot \omega(\sqrt{{\rm log}~m})$. Then multiply $\bold e$ to the both sides of the above equation, we have $\bold A^{(i')}(\bold R_{C_\textsf{PRF},\boldsymbol{\mu}}^{(i')}-\bold R_{C_\textsf{PRF},\boldsymbol{\mu}'}^{(i')})\bold e=\bold 0 \pmod q$ holds, which means the $(\bold R_{C_\textsf{PRF},\boldsymbol{\mu}}^{(i')}-\bold R_{C_\textsf{PRF},\boldsymbol{\mu}'}^{(i')})\bold e$ is a non-zero vector with all but negligible probability and, therefore, a valid the SIS solution for $\bold A^{(i')}$.
\end{proof}

\begin{claim}\label{Claim4.1}
It's hard for $\mathcal{A}$ to forge a signature by adversarially choosing keys.
\end{claim}
\begin{proof}\renewcommand{\qedsymbol}{}
Note that we allowed the adversary $\mathcal{A}$ has the ability to adversarially choosing keys, one attacking method to exploit that is $\mathcal{A}$ can adversarially provides a ring $\textsf{R}=(\textsf{vk}^1, \textsf{vk}^2, \textsf{vk}^3, \textsf{vk}^4)$, suppose only the $(\textsf{vk}^1, \textsf{vk}^2)$ are honest generated which in the verification keys set $\textsf{S}$. $\mathcal{A}$ will successfully forge a signature by querying the singing oracle if $\bold F_{C_\textsf{PRF},\boldsymbol{\mu}, 1-b}^{(1)} = c_1\bold F_{C_\textsf{PRF},\boldsymbol{\mu}, 1-b}^{(3)}$ and $\bold F_{C_\textsf{PRF},\boldsymbol{\mu}, 1-b}^{(2)} = c_2\bold F_{C_\textsf{PRF},\boldsymbol{\mu}, 1-b}^{(4)}$ holds, where $c_1, c_2$ are some constants. It means that $\mathcal{A}$ found the ring $\textsf{R}$ such that each $\bold A_{C_\textsf{PRF},\boldsymbol{\mu}}^{(i)}=c\bold A_{C_\textsf{PRF},\boldsymbol{\mu}'}^{(i)}$ holds for some constant $c$. Assume the based PRF is secure, with $1/2$ probability that $1=\textsf{PRF}(\bold k^{(i')}, \boldsymbol{\mu}) = \textsf{PRF}(\bold k'^{(i')}, \boldsymbol{\mu})$ holds. In this case, $\bold R_{C_\textsf{PRF},\boldsymbol{\mu}}^{(i')}\neq {{\bold R}'}_{C_\textsf{PRF},\boldsymbol{\mu}}^{(i')}$ and $\bold A^{(i')}(\bold R_{C_\textsf{PRF},\boldsymbol{\mu}}^{(i')}-\bold R_{C_\textsf{PRF},\boldsymbol{\mu}'}^{(i')})\pm(1-c) \bold G=\bold 0 \pmod q$ holds. Then we also can giving a reduction from the {\rm $\textsf{SIS}_{q,\beta}$} assumption as same as in Claim \ref{Claim4}.
\end{proof}

\begin{claim}\label{Claim5}
$\mathcal{A}$ can produce a valid {\rm $\textsf{SIS}_{q,\beta}$} solution with overwhelming probability.
\end{claim}
\begin{proof}\renewcommand{\qedsymbol}{}
We argue that ${\hat {\bold x}^*}$ that $\mathcal{S}$ finally output in the simulation is a valid $\textsf{SIS}_{q,\beta}$ solution in two steps. We first explain ${\hat {\bold x}^*}$ is sufficiently short, note that ${\hat {\bold x}^*}$ is consisted by at most $N$ components ${\bar {\bold x}^*}{^{(i^*)}}=\bigg({\bar {\bold x}_1^*}{^{(i^*)}}+\Big(\bold R_{\bold A_{b^*}}^{(i^*)}-\bold R_{C_\textsf{PRF},\boldsymbol{\mu}^{*}}^{(i^*)}\Big)\cdot {\bar {\bold x}_2^*}{^{(i^*)}}\bigg)$ where ${\bar {\bold x}_1^*}{^{(i^*)}}$ and ${\bar {\bold x}_2^*}{^{(i^*)}}$ follow the distribution $\mathcal{D}_{\mathbb{Z}^{m},\sigma}$. By Lemma \ref{Lemma1}, $\lVert{\bar {\bold x}_1^*}{^{(i^*)}}\rVert,\lVert{\bar {\bold x}_2^*}{^{(i^*)}}\rVert\leq \sigma\sqrt{m}$. By Lemma \ref{Lemma7}, we know the norm bound on $\big\lVert{\bar{\bold R}^{(i^*)}}\big\rVert=\Big\lVert\Big(\bold R_{\bold A_{b^*}}^{(i^*)}-\bold R_{C_\textsf{PRF},\boldsymbol{\mu^*}}^{(i^*)}\Big)\Big\rVert\leq O(\ell^{2c}\cdot m^{3/2})$. Therefore, it requires to set $\beta \geq N\cdot O(\ell^{2c}\cdot m^{3/2})\cdot \sigma\sqrt{2m}$. 

Then we prove ${\hat {\bold x}^*}$ is a non-zero with overwhelming probability. Suppose that the $\big\{{\bar {\bold x}_2^*}{^{(i^*)}}\big\}_{i^*\in[N^*]}=\bold 0$, then for a valid forgery we must have at least one ${\bar {\bold x}_1^*}{^{(i^*)}}\neq\bold 0$ in $\big\{{\bar {\bold x}_1^*}{^{(i^*)}}\big\}_{i^*\in[N^*]}$ and thus ${\hat {\bold x}^*}$ is non-zero. Suppose on the contrary, there exists one ${\bar {\bold x}_2^*}{^{(i^*)}}\neq\bold 0$ in $\big\{{\bar {\bold x}_2^*}{^{(i^*)}}\big\}_{i^*\in[N^*]}$, then we need to argue that the corresponding ${\bar {\bold x}^*}{^{(i^*)}}={\bar {\bold x}_1^*}{^{(i^*)}}+\bar{\bold R}^{(i^*)} \cdot{\bar {\bold x}_2^*}{^{(i^*)}}$ is non-zero with overwhelming probability. Due to we assume ${\bar {\bold x}_2^*}{^{(i^*)}}=\big(x_1,\dots,x_m\big)\neq\bold 0$ which means at least one coordinate of ${\bar {\bold x}_2^*}{^{(i^*)}}$, denote as $x_{o}$ where $o\in[m]$, such that $x_{o}\neq 0$. We write $\bar{\bold R}^{(i^*)}=(\bold r_1|\dots|\bold r_m)$ and so $\bar{\bold R}^{(i^*)}\cdot {\bar {\bold x}_2^*}{^{(i^*)}}=\bold r_o x_o+\sum_{\bar o\in[m]\setminus o}\bold r_{\bar o} x_{\bar o}$. Note that for the fixed message $\mu^*$ on which $\mathcal{A}$ made the forgery, $\bar{\bold R}^{(i^*)}$ (therefore $\bold r_{o}$) depends on the low-norm matrices $(\bold R_{\bold A_{0}^{(i^*)}}, \bold R_{\bold A_{1}^{(i^*)}}), \{\bold R_{\bold B_{j}^{(i^*)}}\}_{j\in[k]}, (\bold R_{\bold C_{0}^{(i^*)}}, \bold R_{\bold C_{1}^{(i^*)}})$ and {PRF} key $\bold k^*$. The information about $x_o$ for $\mathcal{A}$ is from the public matrices in the verification set $\textsf{S}$ that given to the $\mathcal{A}$, note that the {PRF} keys $\bold k^*$ which is not included in $\textsf{S}$. So by the pigeonhole principle there is a (exponentially) large freedom to pick a value to $\bold r_{o}$ which is compatible with $\mathcal{A}$'s view. This completes the proof. 
\end{proof}

\subsection{Anonymity}\label{Anonymity}

We now prove the anonymity of $\textsf{RS}$. 

\begin{theorem}[Anonymity]\label{AnonymityProof}
Let $n$ be a security parameter. The parameters $q,m,\sigma,\sigma',\beta$ are chosen as the Sect. \ref{Correctness}. If the {\rm $\textsf{Trapdoor Indistinguishability}$} property of {\rm $\textsf{BasisRand}$} and {\rm $\textsf{Randomness Indistinguishability}$} property of {\rm $\textsf{ExplainGaussian}$} holds, the {\rm $\textsf{RS}$} scheme is signer-unclaimability.
\end{theorem}

\begin{proof}\renewcommand{\qedsymbol}{}
The proof proceeds in two experiments $\textsf{E}_0$, $\textsf{E}_1$ such that $\textsf{E}_0$ (resp., $\textsf{E}_1$) corresponds to the experiment of \textsf{Anonymity} in Definition \ref{RS} with $b = 0$ (resp., $b = 1$), and such that each experiment is statistically indistinguishable from the one before it. This implies that $\mathcal{A}$ has negligible advantage in distinguishing $\textsf{E}_0$ from $\textsf{E}_1$, as desired.

\noindent$\textsf{E}_0:$ This experiment firstly generate $\textsf{PP}\leftarrow\textsf{Setup}(1^n;\gamma_{\textsf{st}})$, and $\{\textsf{vk}^{(i)},\textsf{sk}^{(i)}\}_{i\in[N]}$ by repeatedly invoking $\textsf{KeyGen}(\gamma_{\textsf{kg}}^{(i)})$, and $\mathcal{A}$ is given $(\textsf{PP}, \textsf{S}=\{\textsf{vk}^{(i)}\}_{i\in[N]})$ and the randomness $(\gamma_{\textsf{st}}, \{\gamma_{\textsf{kg}}^{(i)}\}_{i\in[N]})$. Then $\mathcal{A}$ outputs a tuple $(\boldsymbol{\mu}^*, \textsf{R}^*, \textsf{vk}_0^{*}, \textsf{vk}_1^{*})$ where $\textsf{vk}_0^{*}, \textsf{vk}_1^{*}\in\textsf{S}\bigcap \textsf{R}^*$. Finally, $\mathcal{A}$ is given $\Sigma^*$ and $\gamma_{\textsf{sign}}$ that computed by algorithms {\rm $\textsf{Sign}$} and {\rm $\textsf{ExplainGaussian}$} with $\textsf{sk}_0^{*}$, respectively.

\noindent$\textsf{E}_1:$ This experiment is the same as experiment $\textsf{E}_0$ except that the $\big(\Sigma^*, \gamma_{\textsf{sign}}\big)$ given to $\mathcal{A}$ computed by $\textsf{sk}_1^{*}$.

It remains to show that $\textsf{E}_0$ and $\textsf{E}_1$ are statistically indistinguishable for $\mathcal{A}$, which we do by giving a reduction from the \textsf{Trapdoor Indistinguishability} property of \textsf{BasisRand} and {\rm $\textsf{Randomness Indistinguishability}$} property of {\rm $\textsf{ExplainGaussian}$}.

\noindent\textbf{Reduction.} Suppose $\mathcal{A}$ has non-negligible advantage in distinguishing $\textsf{E}_0$ and $\textsf{E}_1$. We use $\mathcal{A}$ to construct an algorithm $\mathcal{S}$ for the \textsf{Trapdoor Indistinguishability} property of \textsf{BasisRand}.

\noindent$\textbf{Simulating Setup Phase}$. $\mathcal{S}$ generates $(\textsf{PP}, \textsf{S}=\{\textsf{vk}^{(i)}\}_{i\in[N]})$ exactly as in experiments $\textsf{E}_0$ and $\textsf{E}_1$, and gives $(\textsf{PP}, \textsf{S}=\{\textsf{vk}^{(i)}\}_{i\in[N]})$ and the appropriate associated randomness $(\gamma_{\textsf{st}}, \{\gamma_{\textsf{kg}}^{(i)}\}_{i\in[N]})$ to $\mathcal{A}$.

\noindent$\textbf{Challenge}$. $\mathcal{A}$ provides a challenge $(\boldsymbol{\mu}^*, \textsf{R}^*, s_0^*,s_1^*)$ to $\mathcal{S}$. $\mathcal{S}$ does the following to response the challenge:
\begin{itemize}
	\item Parse the message $\boldsymbol{\mu}^*=(\mu_1^*,\dots,\mu_t^*)$, ring $\textsf{R}^*=({\textsf{vk}^*}^{(1)},\dots,{\textsf{vk}^*}^{(N)})$, parse each $\textsf{vk}^{(i^*)}=\big(\bold A^{(i^*)},(\bold A_{0}^{(i^*)},\bold A_{1}^{(i^*)}), \{\bold B_{j}^{(i^*)}\}_{j\in[k]},(\bold C_{0}^{(i^*)},\bold C_{1}^{(i^*)})\big)$.  
	\item Let $N^*=|\textsf{R}^*|$. $\mathcal{S}$ checks if $s_0^*,s_1^*\in[N^*]$, $s_0^*\neq s_1^*$ and {\rm $\textsf{vk}^{(s_0^*)},\textsf{vk}^{(s_1^*)} \in \textsf{S}\cap \textsf{R}^*$}, otherwise $\mathcal{S}$ aborts the simulation. 

	\item Let $b^*=b_0^*=b_1^*$. For $i^*=1$ to $N^*$, compute $\bold A_{C_\textsf{PRF},\boldsymbol{\mu}^*}^{{(i^*)}}=\textsf{Eval}\big(C_\textsf{PRF},\big(\{\bold B_{j}^{{(i^*)}}\}_{j\in[k]},\bold C_{\mu_1^*}^{{(i^*)}},$ $\dots,\bold C_{\mu_t^*}^{{(i^*)}}\big)\big)$ and set the matrix \[\bold F_{C_\textsf{PRF},\boldsymbol{\mu}^*,1-b^*}^{{(i^*)}}=\big[\bold A^{(i^*)} | \bold A_{1-b^*}^{(i^*)}-\bold A_{C_\textsf{PRF},\boldsymbol{\mu}^*}^{{(i^*)}}\big]\]
	
	\item Let ${\bold{F}^*_{1-b^*}}=\big[\bold F_{C_\textsf{PRF},\boldsymbol{\mu}^*,1-b^*}^{(1)} |\dots |\bold F_{C_\textsf{PRF},\boldsymbol{\mu}^*,1-b^*}^{(N^*)}\big]$. Compute $\bold{S}_{{\bold{F}^*_{1-b^*}}}^{(s_0^*)}\leftarrow\textsf{BasisExt}\big(\bold{S}_{\bold A^{(s_{0}^*)}}, \textsf{R}^*\big)$ and $\bold{S}_{{\bold{F}^*_{1-b^*}}}^{(s_1^*)}\leftarrow\textsf{BasisExt}\big(\bold{S}_{\bold A^{(\bar i_{1}^*)}}, \textsf{R}^*\big)$.

	\item Send $\bold{S}_{{\bold{F}^*_{1-b^*}}}^{(s_0^*)}$ and $\bold{S}_{{\bold{F}^*_{1-b^*}}}^{(s_1^*)}$ to the challenger $\mathcal{C}$. Then $\mathcal{C}$ chooses a random bit $b\xleftarrow{\$}\{0,1\}$, responses $\bold{S}'_{{\bold{F}^*_{1-b^*}}}\leftarrow\textsf{BasisRand}\big(\bold{S}_{{\bold{F}^*_{1-b^*}}}^{(s_b^*)}, \sigma\big)$.
	
	\item Compute $\gamma_{\textsf{sign}}\leftarrow${\rm $\textsf{ExplainGaussian}({\bold{F}^*_{1-b^*}},\bold{S}'_{{\bold{F}^*_{1-b^*}}}, \bold x',$ $\sigma', \bold 0)$}, $\bold x'\leftarrow \textsf{SampleGaussian}({\bold{F}^*_{1-b^*}},\bold{S}'_{{\bold{F}^*_{1-b^*}}},\bold 0,\sigma')$ such that $\bold{F}^*_{1-b^*}\cdot \bold x'= \bold 0 \pmod q$.
	
	\item Response $(\bold x', \gamma_{\textsf{sign}})$ to $\mathcal{A}$.
	\end{itemize}
\noindent$\textbf{Guess}$. When $\mathcal{A}$ outputs the guess $b'$, $\mathcal{S}$ outputs the guess $b'$.
	
Note that if the random bit $b$ that challenger selected s.t. $b=0$ then the view of $\mathcal{A}$ is distributed exactly according to experiment $\textsf{E}_0$, while if the random bit $b$ that challenger selected s.t. $b=1$ then the view of $\mathcal{A}$ is distributed exactly according to experiment $\textsf{E}_1$. By the \textsf{Trapdoor Indistinguishability} property of \textsf{BasisRand} (Lemma \ref{PropertyOfBasisRand}) and {\rm $\textsf{Randomness Indistinguishability}$} property of {\rm $\textsf{ExplainGaussian}$} (Lemma \ref{PropertyOfExplainGaussian}), $\textsf{E}_0$ and $\textsf{E}_1$ are statistical indistinguishability. This completes the proof.

\end{proof}

\section{CONCLUSION AND FUTURE WORKS}

In this paper, we present a new lattice-based ring signature scheme with unclaimable anonymity. Particularly, our work simultaneously improves the security and efficiency of the work \cite{PS19}. We proved that the scheme is unforgeable w.r.t. adversarially-chosen-key attack in the standard model based on standard lattice assumptions. The comparison shows that our work is the first lattice-based ring signature scheme with unclaimable anonymity in the standard model and with competitive efficiency. As for future works, it is interesting to focus on how to improve our work with the signature size that is logarithmic in the number of ring members, while at the same time relying on standard lattice assumptions and in the standard model.

\ack{This research was supported by the National Natural Science Foundation of China (61672339) and the Shanghai Aerospace Science and Technology Innovation Foundation (SAST2019-008).}

%\nocite{*}

%\bibliographystyle{compj}
%
%\bibliography{out}

\end{document}